\documentclass{aptpub}
\usepackage{amsmath}
\makeatletter
\renewcommand\@biblabel[1]{}
\makeatother

\authornames{R. R. Wilkinson, F. G. Ball, K. J. Sharkey} 
\shorttitle{The deterministic Kermack-McKendrick model bounds the general stochastic epidemic} 

\begin{document}

\title{The deterministic Kermack-McKendrick model\\bounds the general stochastic epidemic} 

\authorone[The University of Liverpool]{Robert R. Wilkinson} 
\authortwo[The University of Nottingham]{Frank G. Ball}
\authorthree[The University of Liverpool]{Kieran J. Sharkey}

\addressone{Department of Mathematical Sciences, The University of Liverpool, Liverpool L69 7ZL, UK. Email address: robert.wilkinson@liverpool.ac.uk} 
\addresstwo{School of Mathematical Sciences, The University of Nottingham, University Park,
Nottingham  NG7 2RD, UK.  Email address: frank.ball@nottingham.ac.uk}
\addressthree{Department of Mathematical Sciences, The University of Liverpool, Liverpool L69 7ZL, UK. Email address: k.j.sharkey@liverpool.ac.uk}

\begin{abstract}
We prove that, for Poisson transmission and recovery processes, the classic Susceptible $\to$ Infected $\to$ Recovered (SIR) epidemic model of Kermack and McKendrick provides, for any given time $t>0$, a strict lower bound on the expected number of suscpetibles and a strict upper bound on the expected number of recoveries in the general stochastic SIR epidemic. The proof is based on the recent message passing representation of SIR epidemics applied to a complete graph. \end{abstract}



\section{Introduction}
\label{sec:intro}

The so-called general stochastic SIR (Susceptible $\to$ Infected $\to$ Recovered) epidemic (see Bailey 1975, Chapter 6) can be defined as follows. We have a finite population consisting of a set $\mathcal{V}$ of $N=|\mathcal{V}|$ discrete individuals. An individual, while infected, makes infectious contacts to any other given individual according to a Poisson process of rate $\beta$ (all such Poisson processes are independent). If a susceptible individual receives an infectious contact, it immediately becomes infected and remains so for an exponentially distributed period of time, with parameter $\gamma$ (and hence mean $\gamma^{-1}$), after which it ceases making contacts and becomes permanently recovered.  All such infectious periods are independent and also independent of the Poisson processes that govern infectious contacts.

The term general is now a misnomer, since far more complicated epidemic models have been proposed and analysed, but for ease of reference we keep with that terminology.  The model has its origin in McKendrick (1926).

For $t \ge 0$, let $X(t)$ and $Y(t)$ be discrete random variables representing the number of susceptible individuals at time $t$ and the number of infected individuals at time $t$, respectively. Then the continuous-time Markov chain $\{ (X(t),Y(t)):t \ge 0 \}$, with transition rates as in Table \ref{table2} with the constraint that $x+y \le N$ and $x, y \ge 0$, is consistent with the dynamics in the general stochastic SIR epidemic ($x$ and $y$ denote possible values of $X(t)$ and $Y(t)$, respectively). For $t \ge 0$, let the random variable $Z(t)=N-X(t)-Y(t)$ represent the number recovered at time $t$.
 
\begin{table}
	\begin{center}
 	\caption{Population-level transition rates for the general stochastic epidemic}
 	\begin{tabular}{cccc}
 		\hline\noalign{\smallskip}
 		From & To & Rate\\
 		\noalign{\smallskip}\hline\noalign{\smallskip}
 		$(x,y)$ & $(x-1,y+1)$ & $\beta x y$ \\
 		$(x,y)$ & $(x,y-1)$ & $\gamma y$ \\ 
 		\noalign{\smallskip}\hline
 		\end{tabular}
 		\label{table2}
 		\end{center}
\end{table}

We assume a pure (i.e.~non-random) initial state. Thus, we define $\mathcal{S}_0 \subset \mathcal{V}$ and $\mathcal{I}_0 = \mathcal{V} \setminus \mathcal{S}_0$ to be the set of initial susceptibles and the set of initial infectives respectively. In order to avoid triviality we assume that $|\mathcal{I}_0| \ge 1$ and $|\mathcal{S}_0| \ge 1$.  For ease of presentation of the proof, we make the stronger assumption that $|\mathcal{S}_0| \ge 2$, though Theorems~\ref{theorem1} and~\ref{theorem2} still hold when $|\mathcal{S}_0|=1$; see Remark~\ref{remark:S(0)=1} after Theorem~\ref{theorem1}.
 		
The deterministic general SIR epidemic (Kermack and McKendrick 1927), in the case where Poisson infection and recovery processes are assumed, is defined by the following system of ordinary differential equations, which we refer to henceforth as the deterministic SIR system:
\begin{eqnarray} \label{Sdet}
 		\dot {S}(t)&=&-\beta   S(t) I(t),  
 		\\ \label{Idet}
 		\dot {I}(t)&=&\beta  S(t) I(t) - \gamma I(t),  \\  
 		\dot {R}(t)&=& \gamma I(t). 
 		\label{kerm}
\end{eqnarray}
Here, and throughout the paper, we use `dot' notation to denote time derivatives.  This model has its origin in McKendrick (1914); Kermack and McKendrick (1927) treat a more general model in which the recovery and infection rates of an indidvidual may depend on the time since it was infected.

By setting $S(0)=|\mathcal{S}_0|, I(0)=|\mathcal{I}_0|,R(0)=0$ and matching the parameters, the deterministic SIR system approximates the general stochastic epidemic and becomes `exact' in the limit of large population size subject to suitable scaling of the initial conditions (Ethier and Kurtz 1986, Chapter 11; Andersson and Britton 2000, Chapter 5).
 		
When the recovery rate $\gamma=0$, the above models reduce to the simple, or SI (Susceptible $\to$ Infected), stochastic and deterministic epidemics (Bailey 1975, Chapter 5).  For these SI models, the deterministic model underestimates the expected number of susceptibles in the stochastic model at any given time $t>0$ (Bailey 1975, page 46; Ball and Donnelly 1987).  For the Markovian SIS (Susceptible $\to$ Infected $\to$ Susceptible) model, it has also been shown (Allen 2008; Simon and Kiss 2012) that its deterministic counterpart underestimates the expected number susceptible at any time.  A long-standing conjecture is that this comparison holds also for the general stochastic SIR epidemic.  The SI and SIS models are essentially one-dimensional, in that $X(t)+Y(t)=S(t)+I(t)=N$ for all $t \ge 0$.  For these models, the comparison between the stochastic and deterministic counterparts is proved easily by using the Kolmogorov forward equation to express $\dot{{\rm E}}[X(t)]$ as the expectation of a quadratic function of $X(t)$ and then using the fact that ${\rm E}\left[X(t)^2\right] > {\rm E}\left[X(t)\right]^2$ to compare $\dot{{\rm E}}[X(t)]$ with $\dot {S}(t)$.  This approach has proved to be unfruitful for the general stochastic SIR epidemic as that model is not one-dimensional.  
 		
 		The aim of the present paper is to use the recently developed message passing approach to general epidemics on
 		graphs (Karrer and Newman 2010; Wilkinson and Sharkey 2014) to prove the above-mentioned conjecture. This is achieved in section \ref{main}, where we show that for finite populations, the deterministic SIR system strictly underestimates (overestimates) the expected number susceptible (recovered) in the general stochastic epidemic, at any positive time point (Theorems~\ref{theorem1} and~\ref{theorem2}) and also at the end of the epidemic (Corollary~\ref{corollary1}). In section~\ref{disc}, we give some brief concluding comments.

 		\section{The deterministic general SIR epidemic provides rigorous bounds for the general stochastic SIR epidemic}
 		\label{main}
 		
 		
 		We show that the deterministic SIR system, expressed as (\ref{Sdet})-(\ref{kerm}), provides a rigorous lower bound on the expected number susceptible at all time points in the general stochastic SIR epidemic (Table 1). We state this main result in the following theorem.
 		
 		\begin{theorem} For the same initial conditions and parameters,
 			\begin{eqnarray}  \nonumber
 			{\rm E}[X(t)] & > & S(t) \qquad \mbox{for all  } t>0.  
 			\end{eqnarray}
 			\label{theorem1}
 			\end{theorem}

 			\begin{proof} From Lemmas \ref{lemma1} and \ref{lemma2} (below),
 				\begin{equation}  \label{SMDorder}
 				{\rm E}[X(t)] \ge S_{\text{mes}}(t)>S(t) \qquad \mbox{for all  } t>0,  
 				\end{equation}
 				where $S_{\text{mes}}(t)$ is defined by \eqref{original} below.
 				\end{proof}

 				Let $T=\inf\{t>0:Y(t)=0\}$ denote the duration of the general stochastic SIR epidemic, so $|\mathcal{S}_0|-X(T)$ is the size of the epidemic, i.e.~the number of initial susceptibles that are infected by the epidemic. The following corollary shows that the expected size of the general stochastic epidemic is strictly less than the size $S(0)-S(\infty)$ of its deterministic counterpart. 
 				
 				\begin{corollary} For the same initial conditions and parameters,
 					\begin{eqnarray}  \nonumber
 					{\rm E}[X(T)] & > & S(\infty).  
 					\end{eqnarray}
 					\label{corollary1}
 					\end{corollary}
 					
 					\begin{proof}
 						The duration $T$ is bounded above by the sum of the infectious periods of all individuals infected during the epidemic, which is almost surely finite, so $X(T)=\lim_{t \to \infty}X(t)$ almost surely and the dominated convergence theorem yields that ${\rm E}[X(T)]=\lim_{t \to \infty}{\rm E}[X(t)]$.  Letting $t \to \infty$ in Theorem~\ref{theorem1} implies immediately that ${\rm E}[X(T)]\ge S(\infty)$.  
 						Moreover, in view of~\eqref{SMDorder}, the inequality is strict, since $S_{\text{mes}}(\infty)>S(\infty)$, see Remark~\ref{remark:Sinfty} after Lemma~\ref{lemma2}.
 						\end{proof} 
 						
 						In the following subsections we present two systems which approximate the general stochastic SIR epidemic, and show that, for any $t>0$, they increasingly underestimate the expected number of susceptibles ${\rm E}[X(t)]$ and increasingly overestimate the expected number recovered ${\rm E}[Z(t)]$. The inequalities are stated in Lemmas \ref{lemma1} and \ref{lemma2}, and Theorem \ref{theorem2}.

 						\subsection{The message passing system}
 						
 						Here we form an approximating system for the general stochastic SIR epidemic from the message passing equations of Karrer and Newman (2010), where we make use of the generalisation of initial conditions provided by Wilkinson and Sharkey (2014). It has been shown that a message passing approach can be applied to a general class of SIR epidemics on finite graphs (Karrer and Newman 2010). If the underlying graph is a tree or forest, then the equations give results which match expectations of the stochastic process exactly; otherwise they give a lower bound on the probability of any given individual being susceptible at time $t$, and hence also a lower bound on the expected number susceptible at time $t$. 
 						
 						The general stochastic SIR epidemic is equivalent to a Markovian SIR epidemic on a complete graph. Thus, the message passing equations for the general stochastic SIR epidemic become (Wilkinson and Sharkey 2014; and see appendix A):
 						\begin{equation} S_{\text{mes}}(t)= \sum_{i \in \mathcal{V}} z_i \prod_{j \neq i} F^{i \leftarrow j}(t), \label{original}  \end{equation}
 						where $z_i=1$ if $i \in \mathcal{V}$ is initially susceptible and is zero otherwise, and
 						\begin{equation} \label{firstF} F^{i \leftarrow j}(t)=1 -  \int_0^t \beta {\rm e}^{-(\beta + \gamma)\tau} \Big[ 1  -   z_j \prod_{k \neq i,j} F^{j \leftarrow k}(t-\tau)   \Big] \mbox{d} \tau \quad (i,j \in \mathcal{V}, i \ne j). \end{equation}
 						The function $S_{\text{mes}}(t)$ is constructed to approximate the expected number of susceptibles at time $t$, while $F^{i \leftarrow j}(t)$ approximates the probability that $i \in \mathcal{V}$ (in the cavity state; see appendix A) does not receive any infectious contacts from $j \in \mathcal{V}$ by time $t$.
 						
 						Note that if $i \notin \mathcal{S}_0$, then $z_i=0$ and (\ref{original}) can be expressed as
 						\[
 						S_{\text{mes}}(t)=\sum_{i \in \mathcal{S}_0} \prod_{j \neq i} F^{i \leftarrow j}(t).
 						\]
 						By symmetry, $F^{i \leftarrow j}(t)=F^{i' \leftarrow j'}(t)$ for any $i,j,i',j' \in \mathcal{S}_0$ with $i \ne j$ and $i' \ne j'$ (let $F_1(t)$ denote this quantity); and, if $j \in \mathcal{I}_0$ so $z_j=0$, then $F^{i \leftarrow j}(t)=1 - \int_0^t \beta {\rm e}^{-(\beta + \gamma)\tau} \mbox{d} \tau$ for any $i \in \mathcal{S}_0$ (let $F_2(t)$ denote this quantity). Thus, the assumed pure initial system state and the symmetry in (\ref{firstF}) allow us to simplify (\ref{original}) to:
 						\begin{eqnarray} 
 						S_{\text{mes}}(t)&=&|\mathcal{S}_0| F_1(t)^{|\mathcal{S}_0|-1}F_2(t)^{|\mathcal{I}_0|},  
 						\label{purestate}
 						\end{eqnarray}
 						where 
 						\begin{eqnarray} \label{F1}
 						F_1(t)& =& 1 - \int_0^t \beta {\rm e}^{-(\beta + \gamma)\tau} \left[ 1 -    F_1(t-\tau)^{| \mathcal{S}_0 |-2}   F_2(t-\tau)^{| \mathcal{I}_0 |}   \right] \mbox{d} \tau, \\
 						F_2(t)&=&\frac{\beta {\rm e}^{-(\beta + \gamma)t}+ \gamma}{\beta + \gamma}. 
 						\label{F2}
 						\end{eqnarray}
 						The message passing system is then completed by
 						\begin{equation} 
 						I_{\text{mes}}(t)= N-S_{\text{mes}}(t)-R_{\text{mes}}(t), \qquad 
 						\dot{R}_{\text{mes}}(t)= \gamma I_{\text{mes}}(t),
 						\label{message}
 						\end{equation}
 						with $I_{\text{mes}}(0)=|\mathcal{I}_0|$ and $R_{\text{mes}}(0)=0$.
 						
 						\begin{lemma}
 							For the same initial conditions and parameters,
 							\begin{eqnarray}
 							\label{XgtSmes}
 							{\rm E}[X(t)] & \ge & S_{\text{mes}}(t) \qquad \mbox{for all  } t>0.  
 							\end{eqnarray}
 							\label{lemma1}
 							\end{lemma}
 							\begin{proof}      Lemma \ref{lemma1} follows directly from the extension of the arguments of Karrer and Newman (2010, section III), provided by Wilkinson and Sharkey (2014, section III) (for further details, see appendix A). 
 								\end{proof}

 									\subsection{The deterministic SIR system}
 									
 									Here we consider the deterministic SIR system given by (\ref{Sdet})-(\ref{kerm}). To proceed, we first reformulate these equations as 
 									\begin{eqnarray} \label{s} 
 									\dot{ s}(t)&=&-\beta s(t) \Big[ |\mathcal{S}_0| i(t) + |\mathcal{I}_0| {\rm e}^{- \gamma t} \Big], 
 									\\  \nonumber  \label{i} 
 									\dot{ i}(t)&=& \beta s(t) \Big[ |\mathcal{S}_0|  i(t) + |\mathcal{I}_0| {\rm e}^{- \gamma t} \Big] - \gamma i(t),  
 									\\   \nonumber
 									\dot{ r}(t)&=& \gamma i(t),  
 									\label{individual}
 									\end{eqnarray}
 									such that, setting $s(0)=1, i(0)=r(0)=0$, we have
 									\begin{eqnarray}  \label{Si2}
 									S(t)&=&|\mathcal{S}_0| s(t), \\   \label{Ii2}
 									I(t) &=& |\mathcal{S}_0| i(t) + |\mathcal{I}_0|{\rm e}^{- \gamma t}, \\  \nonumber
 									R(t)&=&|\mathcal{S}_0| r(t) + |\mathcal{I}_0|(1- {\rm e}^{- \gamma t}).
 									\label{Si}
 									\end{eqnarray}
 									In this form, $s(t), i(t)$ and $r(t)$ represent the state of the initially-susceptible population (i.e.~the fraction that are respectively susceptible, infected and recovered at time $t$), with the initially-infected population becoming recovered at rate $\gamma$.

 									\begin{lemma}
 										For the same initial conditions and parameters,
 										\begin{eqnarray}   \nonumber
 										S_{\text{mes}}(t) & > & S(t) \qquad \mbox{for all  } t>0.  
 										\end{eqnarray}
 										\label{lemma2}
 										\end{lemma}
 										\begin{proof} By analogy with (\ref{purestate}), we first reformulate (\ref{Si2}) in terms of two quantities, $S_1 (t)$ and $S_2(t)$, which are defined such that
 											\begin{eqnarray} 
 											S(t)&=&|\mathcal{S}_0| S_1(t)^{|\mathcal{S}_0|}S_2(t)^{|\mathcal{I}_0|},  
 											\label{laS}
 											\end{eqnarray}
 											where
 											\begin{eqnarray}        \label{S1}
 											\dot{S_1}(t)&=&-\beta S_1(t) i(t), \qquad\;\, S_1(0)=1, \\
 											\dot{S_2}(t)&=&-\beta S_2(t) {\rm e}^{- \gamma t}, \qquad S_2(0)=1.  
 											\label{S2}
 											\end{eqnarray}
 											(Differentiating~\eqref{laS} with respect to $t$ and substituting from~\eqref{Si2},~\eqref{S1} and~\eqref{S2}, shows that~\eqref{s} is satisfied.) 
 											Note that $S_1(t)$ and $S_2(t)$ are strictly decreasing from 1 and greater than 0. Thus, if $F_1(t) > S_1(t)$ and $F_2(t) > S_2(t)$ for all $t > 0$, then Lemma \ref{lemma2} must hold (compare (\ref{purestate}) and (\ref{laS})).

 											Immediately, we can solve~\eqref{S2} for $S_2(t)$, yielding
 											\begin{eqnarray}  \label{S2t}
 											S_2(t)&=&{\rm e}^{\textstyle{\frac{\beta}{\gamma}({\rm e}^{- \gamma t}-1)}},
 											\end{eqnarray}
 											which allows us to alternatively express its time derivative as
 											\begin{eqnarray} 
 											\dot{S_2}(t)&=&-\gamma \Big[ S_2(t) \ln S_2(t) \Big]    - \beta S_2(t). 
 											\label{S2deriv}
 											\end{eqnarray}
 											The time derivative for $F_2(t)$ can be similarly expressed as
 											\begin{equation} \dot{F_2}(t)= -\gamma \Big[ F_2(t)-1\Big] - \beta F_2(t). \label{F2deriv} \end{equation}
 											Thus, since $S_2(t)$ and $F_2(t)$ are strictly decreasing from 1 and greater than 0, and since $0 > x \ln x > x-1$ for $x \in (0,1)$, then (\ref{S2deriv}) and (\ref{F2deriv}) show that $S_2(t)=F_2(t)$ implies $\dot{S_2}(t)< \dot{F_2}(t)$, whence 
 											\begin{equation} \label{F_2>S_2} F_2(t) > S_2(t) \qquad \mbox{for all  } t>0. \end{equation}

 											Dividing (\ref{Idet}) by (\ref{Sdet}) and using separation of variables yields, after invoking the initial conditions
 											$S(0)=|\mathcal{S}_0|$ and $I(0)=|\mathcal{I}_0|$, that
 											\begin{equation} \label{I}
 											I(t)=N-S(t)+ \frac{\gamma}{\beta} \ln \frac{S(t)}{|\mathcal{S}_0|},
 											\end{equation}
 											which on substituting into~\eqref{Sdet} gives
 											\begin{equation}  \nonumber
 											\dot{S}(t)=- \beta S(t)\Bigg[N-S(t)+ \frac{\gamma}{\beta} \ln \frac{S(t)}{|\mathcal{S}_0|} \Bigg]. 
 											\end{equation}
 											We now take (\ref{S1}), and substitute from (\ref{Ii2}) and (\ref{I}) to obtain
 											\begin{eqnarray} \nonumber
 											\dot{S_1}(t)&=&- \frac{\beta S_1(t)}{  |\mathcal{S}_0| } \Bigg[ N- S(t) + \frac{\gamma}{\beta} \ln \frac{S(t)}{|\mathcal{S}_0|} - |\mathcal{I}_0| {\rm e}^{- \gamma t} \Bigg] \\ \nonumber
 											&=& - \frac{\beta S_1(t)}{  |\mathcal{S}_0| }     \Bigg[ N - |\mathcal{S}_0|S_1(t)^{|\mathcal{S}_0|}S_2(t)^{|\mathcal{I}_0|} \\  
 											&&+\frac{\gamma}{\beta} \ln \Big( S_1(t)^{|\mathcal{S}_0|}S_2(t)^{|\mathcal{I}_0|} \Big) - |\mathcal{I}_0| {\rm e}^{- \gamma t}    \Bigg]   ,\label{S1dot}
 											\end{eqnarray}
 											using~\eqref{laS}.
 											Now, $\ln S_2(t)=(\beta / \gamma)({\rm e}^{- \gamma t}-1)$ and $|\mathcal{S}_0|=N-|\mathcal{I}_0|$, so~\eqref{S1dot} simplifies to
 											\begin{eqnarray}
 											\dot{S_1}(t)&=&- \gamma \Big[ S_1(t) \ln S_1(t) \Big] - \beta S_1(t) + \beta \Big[ S_1(t)^{|\mathcal{S}_0|+1} S_2(t)^{|\mathcal{I}_0|} \Big]. 
 											\label{S1deriv}
 											\end{eqnarray}
 											Substituting $u=t-\tau$ in the integral in~\eqref{F1}, so that $t$ may be taken out of the integrand, the time derivative for $F_1(t)$ can be expressed similarly as
 											\begin{eqnarray}
 											\dot{F_1}(t)&=&- \gamma \Big[  F_1(t)-1 \Big] - \beta F_1(t) + \beta \Big[ F_1(t)^{|\mathcal{S}_0|-2} F_2(t)^{|\mathcal{I}_0|} \Big].
 											\label{F1deriv}
 											\end{eqnarray}
 											Note that, since $F_1(0)=1$ and $F_2(t) \in [0,1]$ for all $t \ge 0$, then $F_1(t) \in [0,1]$ for all $t \ge 0$ (consider the possible values of the right-hand side of (\ref{F1deriv}) when $F_1(t)=0,1$). Thus, since (i) $0 > x \ln x > x-1$ for $x \in (0,1)$, (ii) $F_2(t) > S_2(t) \, \, \mbox{for all  } t>0$ and (iii) $|\mathcal{S}_0| \ge 2$, (\ref{S1deriv}) and (\ref{F1deriv}) show that $S_1(t)=F_1(t)$ implies $\dot{S_1}(t)<\dot{F_1}(t)$, whence
 											\begin{equation}  \nonumber
 											F_1(t) > S_1(t) \qquad \mbox{for all  } t>0,
 											\end{equation}
 											and indeed Lemma \ref{lemma2} must hold. 
 											
 											\end{proof}
 											
 											\begin{remark}
 												\label{remark:S(0)=1}
 												Note that when $|\mathcal{S}_0|=1$, ${\rm E}[X(t)]=F_2(t)^{|\mathcal{I}_0|}$ and $S(t)=S_1(t) S_2(t)^{|\mathcal{I}_0|}$, so in this case Theorem~\ref{theorem1} follows immediately from~\eqref{F_2>S_2}, since $S_1(t) \in [0,1]$ for all $t \ge 0$.
 												\end{remark}
 												
 												\begin{remark}
 													\label{remark:Sinfty}
 													Letting $t \to \infty$ in~\eqref{F2} and~\eqref{S2t} yields $F_2(\infty)=\gamma/(\gamma+\beta)$ and $S_2(\infty)={\rm e}^{- \beta / \gamma}$, whence $F_2(\infty)>S_2(\infty)$.  Further, letting $t \to \infty$ in~\eqref{I} shows that $S(\infty)>0$, whence $F_1(\infty)>0$ and $S_1(\infty)>0$.  Letting $t \to \infty$ in~\eqref{purestate} and~\eqref{laS} now shows that
 													$S_{\text{mes}}(\infty) > S(\infty)$, so the size of the message passing epidemic is strictly less than that of the corresponding deterministic SIR epidemic.
 													\end{remark}
 													
 													\begin{remark}
 														It can also be shown, by entirely analogous means, that the stochastic `carrier-borne' epidemic of Downton (1968) is underestimated by its deterministic counterpart, in terms of the expected number susceptible at time $t$. This model is equivalent to the general stochastic SIR epidemic, except that when a susceptible individual receives an infectious contact it becomes infected (a carrier) independently with probability $\pi$, and otherwise becomes immediately recovered. The standard deterministic version, for tracking the number of susceptibles and the number of carriers, is obtained from (\ref{Sdet}) and (\ref{Idet}), but with the first term on the right-hand side of (\ref{Idet}) multiplied by $\pi$. 
 														
 														We note that the probability of an initially susceptible individual still being susceptible at time $t$, in the stochastic carrier-borne model, is the same as in the general stochastic SIR epidemic when it is modified such that every individual which is not initially infected is independently vaccinated with probability $1- \pi$ and initially susceptible otherwise (cf.~Ball (1990)). The message passing equations for such initial conditions were considered by Wilkinson and Sharkey (2014) and shown to provide a lower bound on the expected number susceptible at time $t$. The message passing equations for the carrier-borne epidemic are obtained from (\ref{purestate})-(\ref{message}), but with the integral in (\ref{F1}) multiplied by $\pi$.
 														\end{remark}

 														\subsection{Bounding the expected number recovered}
 														
 														It is now straightforward to show that the deterministic SIR system overestimates the expected number recovered at all positive time points in the general stochastic SIR epidemic. We state this result in the following theorem.

 														\begin{theorem}
 															For the same initial conditions and parameters,
 															\begin{equation}   \nonumber
 															{\rm E}[Z(t)]   <  R(t)   \qquad \mbox{for all  } t>0.   
 															\end{equation}
 															\label{theorem2}
 															\end{theorem}
 															\begin{proof}
 																It is straightforward to show using the Kolmogorov forward equation that for the general stochastic SIR epidemic,
 																\begin{equation}  \nonumber
 																\dot{{\rm E}}[Z(t)]= \gamma {\rm E}[Y(t)],
 																\end{equation}
 																and recall from~\eqref{kerm} that for the deterministic SIR system,
 																\begin{equation}  \nonumber
 																\dot{R}(t)=\gamma I(t).
 																\label{R3}
 																\end{equation}
 																It is also straightforward that
 																\begin{equation}  \nonumber
 																{\rm E}[X(t)]+{\rm E}[Y(t)]+{\rm E}[Z(t)]=N
 																\end{equation}
 																and  
 																\begin{equation}  \nonumber
 																S(t)+I(t)+R(t)=N.
 																\end{equation}
 																Therefore, assuming ${\rm E}[Z(0)]=R(0)=0$,  
 																\begin{eqnarray}
 																{\rm E}[Z(t)]&=& \int_0^t \gamma {\rm e}^{- \gamma \tau}\Big(N-{\rm E}[X(t-\tau)] \Big) \mbox{d} \tau 
 																\label{Z1}
 																\end{eqnarray}
 																and
 																\begin{eqnarray}
 																R(t)&=& \int_0^t \gamma {\rm e}^{- \gamma \tau}\Big(N-S(t-\tau) \Big) \mbox{d} \tau. 
 																\label{Z3}
 																\end{eqnarray}
 																Since we know from Theorem 1 that the expected number susceptible is underestimated by the deterministic SIR system (at all positive time points) then (\ref{Z1}) and (\ref{Z3}) imply that the expected number recovered must be overestimated.  
 																\end{proof}
 																\begin{remark}
 																	A similar argument shows that ${\rm E}[Z(t)] \le R_{\text{mes}}(t) < R(t)$ for all $t>0$.
 																	\end{remark}

 																	\section{Discussion}
 																	\label{disc}
 																	
 																	By applying the the recently developed message passing approach to SIR epidemics to complete graphs (Karrer and Newman (2010)), we have shown that the Kermack-McKendrick SIR model with Poisson transmission and recovery processes produces rigorous bounds for the general stochastic SIR epidemic, as defined in Bailey (1975). Specifically, the deterministic system~\eqref{Sdet}-\eqref{kerm} underestimates the expected number of susceptibles and overestimates the number of recoveries. Equivalent bounds also apply to the `carrier-borne' epidemic model of Downton (1968). 
 																	
 																	Although, at any time $t>0$ and at the end of an epidemic, the mean number of susceptibles in the general stochastic SIR epidemic is strictly larger than the number of suscepibles in the corresponding deterministic epidemic, the law of large numbers for density dependent population processes (Ethier and Kurtz (1986), Theorem 11.2.1) implies that the difference is small, relative to the population size, when both the population and the initial number of infectives are large; the law of large numbers requires that the fraction initially infected tends to a strictly positive number as $N \to \infty$.  (The law of large numbers assumes that the infection rate $\beta$ depends on the population size $N$, say $\beta=\beta_N$, and that $\beta_N N$ tends to a strictly positive finite limit as $N \to \infty$.) 
If instead the initial number of infectives is held fixed and the epidemic is above threshold (i.e.~$\lim_{N \to \infty} \beta_N N > \gamma$), then, in the limit as $N \to \infty$, the deterministic model represents the expected behaviour, after a random time translation, of epidemics that take off  (Barbour and Reinert 2013).  In these circumstances, unless the fixed initial number of infectives is sufficiently large or the epidemic is well above threshold, the deterministic epidemic overestimates appreciably the expected fraction of the population that is ultimately recovered in the stochastic epidemic, even in the limit as $N \to \infty$, since the latter includes a contribution from the non-neglibible proportion of epidemics that do not take off.

The message passing representation falls between the expected behaviour of the general stochastic epidemic and the deterministic SIR system. Specifically, equation \eqref{SMDorder} implies that, for any $t>0$, $S_{\text{mes}}(t)$ gives a closer \emph{approximation} than $S(t)$ to $E[X(t)]$,  with both being underestimates, and that, under the above asymptotic regime, $S_{\text{mes}}(t)/N$ converges to the same deterministic limit as $E[X(t)/N]$.

 																	An interesting development of our work would be to show that the non-Poisson form of the Kermack-McKendrick model also provides bounds on the corresponding stochastic process. Note that such an extension includes SEIR (Susceptible $\to$ Exposed $\to$ Infected $\to$ Recovered) models.  Another extension worthy of investigation is to multitype SIR epidemics.  
 																	
 																	\appendix
 																	
 																	\section{Message passing equations for a class of Markovian epidemics on finite graphs}
 																	
 																	We consider a stochastic SIR epidemic on an undirected simple graph having finite vertex set $\mathcal{V}$.
 																	The disease dynamics are the same as those described in Section~\ref{sec:intro} for the general stochastic epidemic, except now if individual $i \in \mathcal{V}$ becomes infected it makes infectious contacts only to individuals in $\mathcal{N}_i$, the set of neighbours of $i$ in the graph.  The general stochastic epidemic is obtained by taking the graph to be the complete graph on $N$ individuals.
 															                We assume a non-random initial state, in which each individual is initially either susceptible or infected; for $i \in \mathcal{V}$, we set $z_i=1$ if $i$ is initially susceptible and $z_i=0$ otherwise.
 															                We outline below the message passing equations for this model and the proof that they overestimate the expected spread of infection.  The model is a special case of that studied in Wilkinson and Sharkey (2014), which allowed for heterogeneous and non-Markovian individual-level processes, and more general (possibly random) initial conditions.  The message passing approach was developed by Karrer and Newman (2010), within the framework of a model with non-Markovian disease dynamics, in which the initial states of individuals are independent and identically distributed random variables.

 																	Message passing relies on the concept of the `cavity state' in order to simplify calculations. An individual is placed into the cavity state by cancelling its ability to make contacts. This does not affect its fate (it only affects the fates of others). However, it means that the probability of an individual being susceptible at time $t$ is equal to the probability that, when it is in the cavity state, it is initially susceptible and does not receive an infectious contact from any of its neighbours by time $t$. 
 																	
 																	For arbitrary $i \in \mathcal{V}$ and neighbour $j \in \mathcal{N}_i$, let $H^{i \leftarrow j}(t)$ denote the probability that $i$, when in the cavity state, does not receive an infectious contact from $j$ by time $t$. We can now write
 																	\begin{equation} \nonumber
 																	H^{i \leftarrow j}(t)=1- \int_0^t \beta {\rm e}^{-( \beta + \gamma)\tau}  \big(1-z_j \Phi_i^j(t-\tau)  \big) \mbox{d} \tau,
 																	\end{equation} 
 																	where $\Phi_i^j(t)$ is the probability that $j$ does not receive any infectious contacts by time $t$ when $i$ and $j$ are both in the cavity state.  (The probability that $j$ makes an infectious contact to $i$ during the time interval $[t,t+\Delta t)$, where time is measured from the moment $j$ becomes infected, is given by $\beta {\rm e}^{-( \beta + \gamma) t} \Delta t +o(\Delta t)$ as $\Delta t \to 0$.)

 																	By the arguments of Karrer and Newman (2010), and Wilkinson and Sharkey (2014), it can be shown that
 																	\begin{eqnarray}
 																	\label{PS_eq}
 																	P_{S_i}(t)  &\ge& z_i \prod_{j \in \mathcal{N}_i } H^{i \leftarrow j}(t) \qquad (i \in \mathcal{V}), 
 																	\label{PS_ineq} 
 																	\end{eqnarray}
 																	where $P_{S_i}(t)$ is the probability that $i$ is susceptible at time $t$, and 
 																	\begin{eqnarray} \nonumber
 																	H^{i \leftarrow j}(t) & \ge & 1- \int_0^t \beta {\rm e}^{-( \beta + \gamma)\tau} \Big(1-z_j \prod_{k \in \mathcal{N}_j \setminus i} H^{j \leftarrow k}(t-\tau)  \Big) \mbox{d} \tau\qquad (i \in \mathcal{V}, j \in \mathcal{N}_i) . \\
 																	\label{h_eqn}
 																	\end{eqnarray} 
 																	Inequalities (\ref{PS_ineq}) and (\ref{h_eqn}) essentially follow from the fact that an individual receiving no infectious contacts from one subset of its neighbours is positively correlated with it receiving no infectious contacts from a different subset.
 																	
 																	We now state the definition of $F^{i \leftarrow j}(t)$ (cf.~equation (\ref{firstF})), noting that it satisfies equality in (\ref{h_eqn}):
 																	\begin{equation} \label{Fappend}
 																	F^{i \leftarrow j}(t) = 1- \int_0^t \beta {\rm e}^{-( \beta + \gamma)\tau} \Big(1-z_j \prod_{k \in \mathcal{N}_j \setminus i} F^{j \leftarrow k}(t-\tau)  \Big) \mbox{d} \tau \qquad (i \in \mathcal{V}, j \in \mathcal{N}_i).
 																	\end{equation}
 																	The unique solution to (\ref{Fappend}) can be obtained via the system of ordinary differential equations:
 																	\begin{eqnarray} \nonumber
 																	\dot{F}^{i \leftarrow j}(t) &=& \gamma \Big( 1- F^{i \leftarrow j}(t) \Big) - \beta \Big( F^{i \leftarrow j}(t)-  z_j \prod_{k \in \mathcal{N}_j \setminus i } F^{j \leftarrow k}(t) \Big)\qquad (i \in \mathcal{V}, j \in \mathcal{N}_i). 
 																	\end{eqnarray}
 																	
 																	Reproducing an argument from Karrer and Newman (2010), we can also construct the solution to (\ref{Fappend}) as follows.
 																	Let $F_{(0)}^{i \leftarrow j}(t)=H^{i \leftarrow j}(t)$ for all $i \in \mathcal{V}, j \in \mathcal{N}_i$ and all $t \ge 0$, and define the following iterative procedure: for $m=1,2,\ldots$,
 																	\begin{equation*} 
 																	F_{(m)}^{i \leftarrow j}(t) = 1- \int_0^t \beta {\rm e}^{-( \beta + \gamma)\tau} \Big(1-z_j \prod_{k \in \mathcal{N}_j \setminus i} F_{(m-1)}^{j \leftarrow k}(t-\tau)  \Big) \mbox{d} \tau \qquad (i \in \mathcal{V}, j \in \mathcal{N}_i).
 																	\end{equation*}
 																	It is easily shown, using (\ref{h_eqn}), that $H^{i \leftarrow j}(t) \ge F_{(m)}^{i \leftarrow j}(t) \ge F_{(m+1)}^{i \leftarrow j}(t)\ge 1- \int_0^t \beta {\rm e}^{-( \beta + \gamma)\tau} \mbox{d} \tau$, for all $i \in \mathcal{V}, j \in \mathcal{N}_i$ and all $t \ge 0$ and $m=0,1,\dots$, whence $ \lim_{m \to \infty}F_{(m)}^{i \leftarrow j}(t)$ is the solution of~\eqref{Fappend} and
 																	\begin{equation}\label{fin} H^{i \leftarrow j}(t) \ge  F^{i \leftarrow j}(t) \ge 1- \int_0^t \beta {\rm e}^{-( \beta + \gamma)\tau} \mbox{d} \tau\qquad (i \in \mathcal{V}, j \in \mathcal{N}_i). \end{equation}

 																	Thus from~\eqref{PS_ineq} and (\ref{fin}) we have
 																	\begin{equation}
 																	\label{PS_ibound}
 																	P_{S_i}(t)  \ge z_i \prod_{j \in \mathcal{N}_i } H^{i \leftarrow j}(t) \ge z_i \prod_{j \in \mathcal{N}_i } F^{i \leftarrow j}(t)  \qquad
 																	(i \in \mathcal{V})
 																	\end{equation}
 																	
 																	Let $X(t)$ denote the number of susceptible individuals at time $t$.  Then ${\rm E}[X(t)]=\sum_{i \in \mathcal{V}} P_{S_i}(t) $, so, recalling~\eqref{original},~\eqref{PS_ibound} implies~\eqref{XgtSmes}.

         \acks
         R.R.W. acknowledges support from EPSRC (DTA studentship). R.R.W. and K.J.S. acknowledge support from the Leverhulme Trust (RPG-2014-341).


        \end{document}